\documentclass[11pt,a4paper]{amsart}
\usepackage{a4wide}
\usepackage{epsfig}
\usepackage{amsmath}
\usepackage{graphicx}
\usepackage{verbatim}
\usepackage{amsmath}
\usepackage{latexsym}
\usepackage{amssymb}
\usepackage{amscd}
\usepackage{hyperref}
\usepackage{amsthm}
\usepackage[latin1]{inputenc}
\usepackage{tikz}

\newcommand{\CC}{\mathbb C}

\newcommand{\ZZ}{\mathbb Z}

\newcommand{\LL}{\mathbb L}

\newcommand{\cL}{\mathcal L}
\newcommand{\cM}{\mathcal M}

\newcommand{\cW}{\mathcal W}

\newtheorem{theorem}{Theorem}
\newtheorem{prop}[theorem]{Proposition}
\newtheorem{lemma}[theorem]{Lemma}

\newtheorem{example}[theorem]{Example}

\newtheorem{Def}[theorem]{Definition}

\begin{document}

\title[Polynomial Modular Frobenius Manifolds]
{Polynomial Modular Frobenius Manifolds}

\author{Ewan K. Morrison, Ian A.B. Strachan}
\date{17$^{\rm th}$ Oct. 2011}
\address{School of Mathematics and Statistics\\ University of Glasgow\\
Glasgow G12 8QQ\\ U.K.}

\email{e.morrison@maths.gla.ac.uk, i.strachan@maths.gla.ac.uk}


\begin{abstract}
The moduli space of Frobenius manifolds carries a natural involutive symmetry, and a distinguished class - so-called modular Frobenius manifolds - lie at the
fixed points of this symmetry. In this paper a classification of semi-simple modular Frobenius manifolds which are polynomial in all but one of the variables
is begun, and completed for three and four dimensional manifolds. The resulting examples may also be obtained from higher dimensional manifolds by a process
of folding. The relationship of these results with orbifold quantum cohomology is also discussed.
\end{abstract}

\maketitle


%
\section{Introduction}

The simplest family of Frobenius manifolds are those for which the prepotential - a solution of the underlying
WDVV-equations - is a polynomial in the flat variables and the Saito construction \cite{saito} provides a construction of these polynomial solutions on the
orbit space $\CC^N/W\,,$ starting
from any finite Coxeter group $W\,.$ It was conjectured by Dubrovin \cite{dubrovin1} and later proved by Hertling \cite{Hertling} that all
polynomial solutions arise in this way, and hence polynomial Frobenius manifolds have now been completely classified.
Beyond polynomial solutions there are trigonometric and elliptic solutions and variations of the original Saito
construction, corresponding to extended affine Weyl and Jacobi groups, which provide large classes of examples. One may conjecture that these constructions provide
all such solutions within their classes, but currently, unlike in the polynomial case, such conjectures remain unproved.

In \cite{MS} the authors introduced the idea of a modular Frobenius manifold.
These sit at the fixed points of a natural involutive symmetry that exists on the
moduli space of Frobenius manifolds and hence inherit very special properties, basically an
invariance under the modular group.

The aim of this paper is to study such modular Frobenius manifolds further and in particular to
begin a low-dimensional classification of polynomial examples, where this means that the prepotential
is polynomial in the variables $t^1\,,\ldots\,,t^{N-1}$ and transcendental in the last variable $t^N$
(throughout this paper $N$ will denote the dimension of the Frobenius manifold).
A number of explicit examples have appeared in the literature before:

\begin{itemize}
	\item[$\bullet$] the $A_1$ Jacobi group example \cite{dubrovin1};
	\item[$\bullet$] the $D_4^{(1,1)}$ solution \cite{SatakeD4};
	\item[$\bullet$] the simple elliptic singularities ${\widetilde{E}}_{6,7,8}$ \cite{KTS,NY,saito3, SatakeFrob,VW}.
	
\end{itemize}

\noindent What is perhaps surprising
is that the simplest family of examples - those based on the $A_{N-2}$-Jacobi group - do not fall
into this class, apart from the simplest $N=3$ case.
In general these are polynomial in the variables $t^1\,,\ldots t^{N-2}$ but have rational dependence in the variable
$t^{N-1}$ \cite{B}.
Thus polynomial modular Frobenius manifolds seem to be extremely special.

\noindent Some other solutions may be found in \cite{MartiniPost} and in the work of Bertola \cite{B}\,, and an aim of this paper is to
provide a systematic analysis of the solution space of polynomial modular manifolds rather than just isolated examples.
In fact a close examination of Satake's construction \cite{SatakeFrob} (itself based on the Saito construction \cite{saito3})
for the so-called codimension one cases shows that the prepotential will be polynomial and modular. Since
there are only a finite number of codimension one examples, a plausible conjecture would be that semi-simple polynomial
modular manifolds are actually finite in number.

\section{Modular Frobenius manifolds}\label{pre}

We recall the basic definitions and symmetries of a Frobenius manifold \cite{dubrovin1}.
\subsection{Frobenius Manifolds and the WDVV Equations}
Frobenius manifolds were introduced as a way to give a geometric understanding to solutions of the Witten-Dijkraaf-Verlinde-Verlinde (WDVV) equations,
\begin{equation}\label{wdvv}
\frac{\partial^3 F}{\partial t^{\alpha} \partial t^{\beta} \partial t^{\lambda}}\eta^{\lambda \mu} \frac{\partial^3 F}{\partial t^{\mu} \partial t^{\delta} \partial t^{\gamma}} =
\frac{\partial^3 F}{\partial t^{\delta} \partial t^{\beta} \partial t^{\lambda}}\eta^{\lambda \mu} \frac{\partial^3 F}{\partial t^{\mu} \partial t^{\alpha} \partial t^{\gamma}}
\end{equation}
for some quasihomogeneous function $F(t)$. Throughout this paper $\eta^{\alpha\beta}$ will be defined via $\eta^{\alpha\beta}\eta_{\beta\kappa}=\delta^{\alpha}_{\kappa}$ where
\begin{equation}
\eta_{\alpha\beta} = \frac{\partial^3 F}{\partial t^1 \partial t^{\alpha} \partial t^{\beta}}
\end{equation}
is constant and non-degenerate.
We recall briefly how to establish the correspondence between Frobenius manifolds and solutions of WDVV.
\begin{Def} The triple $(A,\circ, \eta)$ is a \emph{Frobenius algebra} if:
\begin{itemize}
\item[1.]{$(A, \circ)$ is a commutative associative algebra over $\CC$ with unity $e$;}
\item[2.]{The bilinear pairing $\eta$ and multiplication $\circ$ satisfy the following \emph{Frobenius condition}
\begin{equation*}
\eta( X\circ Y, Z ) = \eta( X, Y \circ Z )\,, \quad X, Y, Z \in A.
\end{equation*}}
\end{itemize}
\end{Def}
\noindent With this one may define a Frobenius manifold.
\begin{Def} Let $\cM$ be a smooth manifold. $\cM$ is called a \emph{Frobenius manifold} if each tangent space $T_t\cM$ is equipped with the structure of a Frobenius algebra varying
smoothly with $t\in \cM$, and further
\begin{itemize}
\item[1.]{The invariant inner product $\eta$ defines a flat metric on $\cM$. }
\item[2.]{The unity vector field is covariantly constant with respect to the Levi-Civita connection for $\eta$,
\begin{equation}\label{fm1}
\nabla e =0.
\end{equation}}
\item[3.]{Let
\begin{equation}
c(X,Y,Z):=\eta( X\circ Y, Z)\,, \quad X, Y, Z \in T_t\cM,
\end{equation}
then the $(0,4)$ tensor $\nabla_Wc(X,Y,Z)$ is totally symmetric. }
\item[4.]{There exists a vector field $E\in\Gamma(T\cM)$ such that $\nabla\nabla E=0$ and
\begin{equation}
\cL_E \eta = (2-d)\eta, \quad \cL_E \circ = \circ, \quad \cL_E e = -e.
\end{equation}}
$E$ is called the Euler vector field.
\end{itemize}
\end{Def}
Condition 1 implies there exists a choice of coordinates $(t^1, ..., t^N)$ such that the Gram matrix $\eta_{\alpha\beta}=( \partial_{\alpha}, \partial_{\beta} )$ is constant.
Furthermore, this may be done in such a way that $e=\partial_1$. In such a coordinate system, partial and covariant derivatives coincide, and condition 3
becomes $c_{\alpha\beta\gamma,\kappa} =c_{\alpha\beta\kappa,\gamma}$. Successive applications of the Poincar\'{e} lemma then implies local existence of a function $F(t)$ called
the \emph{free energy} of the Frobenius manifold such that
\begin{equation}
c_{\alpha\beta\gamma} = \frac{\partial^3 F(t)}{\partial t^{\alpha}\partial t^{\beta} \partial t^{\gamma}}.
\end{equation}
Since $\eta(X,Y) = \eta (e\circ X, Y) = c(e, X, Y)$, we have
\begin{equation}
\eta_{\alpha\beta} = c_{1\alpha\beta}.
\end{equation}
Defining $(\eta_{\alpha\beta})^{-1}=\eta^{\alpha\beta}$, the components of $\circ$ are given by $c^{\alpha}_{\beta\gamma}=\eta^{\alpha\varepsilon}c_{\varepsilon\beta\gamma}$.
Associativity of $\circ$ is then equivalent to \eqref{wdvv}. We assume from now on that $\nabla E$ is diagonalizable and that $\eta_{ij}=\delta_{i+j,N+1}\,.$ Condition 4 leads to the
requirement that $F$ is quasihomogeneous, that is,
\begin{equation}
\cL_EF = E(F) = d_F \cdot F = (3-d)F \quad \mbox{modulo quadratic terms},
\end{equation}
and, on using the freedom in the definition of flat coordinates, the Euler vector field may be taken to be
\begin{equation}\label{euvect1}
E(t) = \sum_{\alpha}(1-\frac{d}{2} - \mu_\alpha)t^{\alpha}\partial_{\alpha} + \sum_{\{\alpha:q_\alpha=1\}} r^\alpha \partial_\alpha.
\end{equation}
We will, for the most part, restrict ourselves in this paper to those Frobenius manifolds with the property that $r^\sigma=0\,,\sigma=1\,,\ldots\,,N\,.$ So
\begin{eqnarray}
E& =&\sum_{\alpha}(1-\frac{d}{2} - \mu_\alpha)t^{\alpha}\partial_{\alpha}\,,\label{enohat}\\
\cL_EF &=&(3-d) F\,.
\end{eqnarray}
This excludes those examples coming from quantum cohomology and the extended affine Weyl orbit spaces \cite{DZ}.

To explain the origins of modular Frobenius manifolds it is first necessary to consider symmetries on the
moduli space of Frobenius manifolds \cite{dubrovin1}.

\begin{Def}\label{symdef}
{}
\begin{itemize}

\item[(a)]\emph{Legendre-type transformations $S_\kappa$}: This is defined by
\begin{eqnarray*}
{\hat t}_\alpha & = & \partial_{t^\alpha} \partial_{t^\kappa} F(t)\,,
\qquad ({\rm note:}\quad{\hat t}_\alpha=\eta_{\alpha\beta}{\hat t}^\beta)\\
\frac{\partial^2{\hat F}}{\partial{\hat t}^\alpha \partial{\hat t}^\beta} & = &
\frac{\partial^2{     F}}{\partial{     t}^\alpha \partial{     t}^\beta} \,,\\
{\hat\eta}_{\alpha\beta} & = & \eta_{\alpha\beta}\,.
\end{eqnarray*}

\medskip

\item[(b)]\emph{Inversion Symmetry $I$:} This is defined by
\begin{displaymath}
\hat{t}^1  = \frac{1}{2}\frac{t_{\sigma}t^{\sigma}}{t^N}, \quad \hat{t}^{\alpha} = \frac{t^{\alpha}}{t^N}, \quad\mbox{  (for }\alpha \neq 1, N), \quad \hat{t}^N = -\frac{1}{t^N},
\end{displaymath}
\begin{equation}\label{isdef}
\hat{F}(\hat{t}) = (\hat{t}^N)^2 F \left( \frac{1}{2}\frac{\hat{t}_{\sigma}\hat{t}^{\sigma}}{\hat{t}^N}, -\frac{\hat{t}^{2}}{\hat{t}^N}, ... , -\frac{\hat{t}^{N-1}}{\hat{t}^N},
-\frac{1}{\hat{t}^N} \right) + \frac{1}{2}\hat{t}^1\hat{t}_{\sigma}\hat{t}^{\sigma},
\end{equation}
\begin{displaymath}
\hat{\eta}_{\alpha\beta} = \eta_{\alpha\beta}.
\end{displaymath}
\end{itemize}
\end{Def}

\medskip

\noindent In this paper we will concentrate on inversion symmetries (also called type II symmetries). For these,
the
structure constants of the inverted Frobenius manifold are related to those of the original by
\begin{equation}\label{ctrans1}
\hat{c}_{\alpha\beta\gamma} =
(t^N)^{-2}\frac{\partial t^{\lambda}}{\partial \hat{t}^{\alpha}}\frac{\partial t^{\mu}}{\partial \hat{t}^{\beta}}\frac{\partial t^{\nu}}{\partial \hat{t}^{\gamma}}c_{\lambda\mu\nu}.
\end{equation}

\noindent Under the assumption that $r_i=0$ for $i=1\,,\ldots\,,N$
the corresponding Euler vector field of the inverted Frobenius manifold $\hat{F}$ has the form
\begin{equation}
\hat{E}(\hat{t}) = \sum_{\alpha}(1-\frac{\hat{d}}{2}-\hat{\mu}_{\alpha})\hat{t}^{\alpha}\hat{\partial}_{\alpha},
\label{ehat}
\end{equation}
where
\begin{equation}
\hat{d}=2-d, \quad \hat{\mu}_1 = \mu_N - 1, \quad \hat{\mu}_N = \mu_1 + 1, \quad \hat{\mu}_i = \mu_i, \quad \mbox{for }i \neq 1, N.
\label{dhat}
\end{equation}

To motivate what follows, consider the following example (taken from \cite{dubrovin1} Appendix C).

\begin{example}\label{exampleChazy}
Suppose $N=3$ and let the Euler vector field and prepotential take the form
\begin{equation}\label{A1hat}
E = t^1\frac{\partial}{\partial t^1} + \frac{1}{2}t^2\frac{\partial}{\partial t^2}, \qquad F(t^1,t^2,t^3)=\frac{1}{2}(t^1)^2t^3 + \frac{1}{2}t^1(t^2)^2-\frac{(t^2)^4}{16}\gamma(t^3).
\end{equation}
The WDVV equations are equivalent to the nonlinear differential equation (known as the Chazy equation)
\begin{equation}\label{chazy1}
\gamma ''' = 6\gamma\gamma'' - 9(\gamma')^2
\end{equation}
for the function $\gamma(t^3)$ (here $'$ denotes differentiation with respect to the variable $t^3$).
The equation has an $SL(2,\CC)$ invariance, mapping solutions to solutions,
\begin{equation}\label{chazymodprop}
t^3 \mapsto {\hat t}^3 = \frac{at^3 + b}{ct^3 + d}, \quad \gamma(t^3) \mapsto {\hat\gamma}({\hat t}^3) =  (ct^3 + d)^2 \gamma (t^3) + 2c(ct^3+d), \quad ad-bc=1\,.
\end{equation}

\medskip

\noindent A simple calculation shows that under the inversion symmetry
\begin{eqnarray*}
{\hat F}({\hat t}) & = &\frac{1}{2}({\hat t}^1)^2{\hat t}^3 + \frac{1}{2}{\hat t}^1({\hat t}^2)^2-\frac{({\hat t}^2)^4}{16}  \left\{ \frac{1}{({\hat t}^3)^2}\gamma\left( \frac{-1}{{\hat t}^3}\right) -
\frac{2}{{\hat t}^3} \right\}\,,\\
& = & \frac{1}{2}({\hat t}^1)^2{\hat t}^3 + \frac{1}{2}{\hat t}^1({\hat t}^2)^2-\frac{({\hat t}^2)^4}{16} {\hat\gamma}({\hat t}^3)
\end{eqnarray*}
where, on using the symmetry given by (\ref{chazymodprop}), ${\hat\gamma}$ is also a solution of the Chazy equation.
Thus under inversion symmetry one has a weak modular symmetry: the functional form of the prepotential is preserved, but with the function $\gamma$ being replaced by a
new solution of the same equation connected by a special case of the symmetry (\ref{chazymodprop}).

\medskip

One simple explicit solution of the Chazy equation is given by
\[
\gamma(t^3) = \frac{\pi i}{3} E_2(t^3)\,,
\]
where $E_2$ is the second Eisenstein series. With this explicit solution one has a stronger symmetry. Since $E_2$ has the modularity properties
\[
E_2(\tau+1) = E_2(\tau)\,, \quad E_2\left(-\frac{1}{\tau}\right) = \tau^2 E_2(\tau) + \frac{12}{2 \pi i}
\]
then under inversion symmetry one has ${\hat F}({\hat t}) = F({\hat t})\,,$ i.e. both the functional form of the prepotential and the specific solution of the Chazy equation
are preserved.

\end{example}

\noindent The above example provides a motivation for studying those Frobenius manifolds with similar properties.

\subsection{Modular Frobenius Manifolds}

\begin{Def} Let $M$ and $\hat{M}$ be two Frobenius manifolds with fixed coordinates $t$ and ${\hat{t}}\,,$ with Euler vector fields $E$ and ${\hat{E}}$ given by (\ref{enohat}) and (\ref{ehat})\,,
and connected via the inversion symmetry given by Definition \ref{symdef}, part (b).

\medskip

The manifold $M$ is defined to be a {\rm modular Frobenius manifold} if
\begin{equation}
{\hat{F}}({\hat{t}}) = F({\hat{t}})\label{invariantF}
\end{equation}
and
\[
{\hat{E}}({\hat{t}}) = E({\hat{t}})\,.
\]
\end{Def}

\noindent Thus instead of transforming the original prepotential into a different prepotential, the inversion symmetry maps the prepotential to itself, as in Example \ref{exampleChazy}, and hence one
may think of modular Frobenius manifolds as lying at the fixed point of the involutive symmetry $I\,.$ A comparison of the two Euler fields, using (\ref{ehat}) and (\ref{dhat}) gives the following
constraints:  ${\hat{d}}=d$ and hence $d=1\,.$ Thus we have the following necessary conditions for a Frobenius manifold to be modular:
\begin{eqnarray}
E& =&\sum_{\alpha=1}^{N-1} (\frac{1}{2} - \mu_\alpha)t^{\alpha}\partial_{\alpha}\,,\\
\cL_EF &=&2 F\,.
\end{eqnarray}

Since the variables $t^1$ and $t^N$ behave differently to the remaining variables it is useful to use a different notation, namely:
\[
u = t^1\,,\qquad z^i = t^i\,,i=2\,,\ldots\,,N-1\,,\qquad \tau = t^N\,,
\]
and ${\bf z} = ( z^2\,,\ldots\,,z^{N-1} )\,.$ The two notations will be used interchangeably in what follows.

A simple calculation then yields the modularity properties of the non-cubic part of the prepotential.
\begin{prop}\label{basicexample}
The prepotential
\begin{equation}
F=\frac{1}{2} u^2 \tau - \frac{1}{2} u ({\bf z},{\bf z}) + f({\bf z},\tau)
\label{prepotentialform}
\end{equation}
(where $({\bf x},{\bf y})=\sum_{i,j=2}^{N-1} \eta_{ij} x^i y^j$) satisfies \eqref{invariantF} under the action of $I$ if and only if
\begin{equation}
f\left(\frac{{\bf z}}{\tau},-\frac{1}{\tau}\right) = \frac{1}{\tau^2} f({\bf z},\tau) - \frac{1}{4 \tau^3} ({\bf z},{\bf z})^2\,.
\label{transformationlaw}
\end{equation}
\end{prop}

\textbf{Remark}. There exists another choice of coordinates on a semi-simple Frobenius manifold, namely one that simplifies its algebraic structure. In these coordinates the multiplication is trivial:
\[
\frac{\partial~}{\partial u^i} \circ \frac{\partial~}{\partial u^j} = \delta_{ij} \frac{\partial~}{\partial u^i}\,.
\]
It turns out that these coordinates are the roots of
\begin{equation}\label{canondef}
\det\left(g^{\alpha\beta}(t)-\lambda \,\eta^{\alpha\beta}\right)=0\,,
\end{equation}
where $g$ is the intersection form of the Frobenius manifold.  The roots of the expression \eqref{canondef} are invariant under the symmetry $I$, and so the canonical
coordinates are preserved up to a re-ordering \cite{dubrovin1}.
This also provides a way to check whether the algebras are nilpotent or semi-simple: one just has to calculate the discriminant of this polynomial to check whether it has
multiple roots (the nilpotent case) or not (the semi-simple case).

\section{Modularity and quasi-homogeneity}

We will consider polynomial prepotentials of the form \eqref{prepotentialform} with
\begin{equation}
f= \sum_{{\bf \alpha}\in \LL} \left\{ \prod_{i=2}^{N-1} (t^i)^{\alpha_i} \right\} g_{\bf \alpha}(t^N)\,,
\label{polyf}
\end{equation}
so ${\bf z}=\{t^2\,,\ldots\,,t^{N-1}\}\,,\tau=t^N$ and $\LL:=\ZZ_{\geq 0}^{N-2}\,.$ We do not assume any general properties of the functions $g_{\bf \alpha}$ at this stage - the aim is to obtain
the differential equations that these functions must satisfy in order for the WDVV equations to hold. Boundary or other conditions (for example, solutions being analytic at infinity)
will then place constraints on these functions, but a priori no such constraints will be imposed.  The following definition,
while not immediately obvious, will play an important role in what follows.

\begin{Def} A pivot-point ${\bf \alpha}\in \LL$ is a lattice point for which
\[
{\rm coeff}_{\bf \alpha} \left( \sum_{i\,,j=2}^{N-1} \eta_{ij} t^i t^j \right)^2 \neq 0\,,
\]
where ${\rm coeff}_{\mathbf \alpha}(p)$ is the coefficient of $\prod_{i=2}^{N-2} (t^i)^{\alpha_i}$ of the polynomial $p\,.$
\end{Def}

\begin{prop} Assume that $f$ defined above satisfies the modularity condition \eqref{transformationlaw}.
Then:

\begin{itemize}

\item[$\bullet$] If $\bf{\alpha}$ is not a pivot-point then:
\[
g_{\bf\alpha}\left(-\frac{1}{\tau}\right) = \tau^{ (\sum \alpha_i)-2 } g_{\bf \alpha}(\tau);
\]

\item[$\bullet$] If $\bf{\alpha}$ is a pivot-point then:
\[
g_{\bf\alpha}\left(-\frac{1}{\tau}\right) = \tau^2 g_{\bf \alpha}(\tau) - \frac{1}{4} \tau  \,
{\rm coeff}_{\bf \alpha} \left( \sum_{i\,,j=2}^{N-1} \eta_{ij} t^i t^j \right)^2
\]

\end{itemize}

\end{prop}

\noindent The proof is by direct computation and will be omitted.

Having determined the modularity properties of the functions we now apply quasi-homogeneity. This will determine the
possible terms that can appear in the ansatz (\ref{polyf}). Recall that the Euler field takes the form
\[
E=\sum_{i=1}^{N-1} d_i t^i \frac{\partial\phantom{t^i}}{\partial t^i}
\]
where $d_i + d_{N+1-i}=1\,.$ We now assume further that the $d_i$ are positive rational numbers. Applying the quasihomogeneity
condition $E(F)=2F$ (recall modular Frobenius manifold must have $d=1$) implies the following constraint on the $\alpha_i\,:$
\begin{equation}
\sum_{i=2}^{N-1} d_i \alpha_i = 2\,.
\label{FrobeniusEquation}
\end{equation}
Thus given the $d_i$ we arrive at a special Diophantine equation whose solutions determines the possible monomials in (\ref{polyf})
(this special type of Diophantine equation is known in number theory as a Frobenius equation).

Pivot-points - which are defined without reference to the Euler vector field - play an important role in the construction of solutions of this equation.

\begin{lemma}

Pivot-points automatically satisfy the Frobenius equation \eqref{FrobeniusEquation}.

\end{lemma}

\begin{proof}
On expanding $\left( \sum_{i\,,j=2}^{N-1} \eta_{ij} t^i t^j \right)^2$ one can obtain the form of the
pivot-points, namely (and here only non-zero elements are shown, all other elements are zero):

\begin{itemize}

\item[(i)] ${\bf \alpha} = ( \ldots,2,\ldots,2,\ldots)$ with $2$ in the $i$ and $N+1-i$ positions;

\item[(ii)] ${\bf \alpha} = ( \ldots,1,\ldots,1,\ldots,1,\ldots,1,\ldots)$ with $1$ in the $i\,,j$ and $N+1-i\,,N+1-j$ positions $(i\neq j)$.

\end{itemize}

\noindent If $N$ is odd one obtains two further pivot-points:

\begin{itemize}

\item[(iii)] ${\bf \alpha} = ( \ldots,4,\ldots)$ with $4$ in the middle position;

\item[(iv)] ${\bf \alpha} = ( \ldots,1,\ldots,2,\ldots,1,\ldots)$ with $1$ in the $i$ and $N+1-i$ positions and $2$ in the middle position.

\end{itemize}

\noindent Since $d_i + d_{N+1-i}=1$ (and hence $d_{(N+1)/2}=1/2$ if $N$ is odd) the result follows.

\noindent Given the above forms of the pivot-points one can easily count the number of such points: if $N$ is even the number of pivot-points is $N(N-2)/8$ and if $N$ is odd the
number of pivot-points is
$(N+1)(N-1)/8\,.$

\end{proof}

One can now give the geometric reasoning behind the name \lq pivot-point\rq. The Frobenius equation (\ref{FrobeniusEquation}) defines a hyperplane $\Pi$ in $\alpha$-space and we
wish to find integer solutions, i.e. the points in $\Pi \cap \LL\,.$ Since pivot-points are independent of
the $d_i\,$ as the $d_i$ vary the plane $\Pi$ \lq pivots\rq~around these points. In the simplest non-trivial example, when $N=4\,,$ we have
$(1-d_3) \alpha_2 +d_3 \alpha_3=2$ and a single pivot-point $(2,2)\,.$ Thus as $d_3$ varies the line rotates, or pivots, about this point:
\begin{center}
\begin{tikzpicture}[domain=0:4]
    \draw[very thin,color=gray] (-0.2,-0.2) grid (6.2,4.2);
    \draw[->] (-0.2,0) -- (6.2,0) node[right] {$\alpha_2$};
    \draw[->] (0,-0.2) -- (0,4.2) node[above] {$\alpha_3$};
    \draw[color=blue][-] (0,4)--(4,0);
    \draw[color=blue](2.9,0.2) node{$d_3=1/2$};
    \draw[color=red][-](0,3)--(6,0);
    \draw[color=red](4.9,1.2) node{$d_3=2/3$};
    \draw[color=black](3.5,3.5) node {$(1-d_3) \alpha_2 +d_3 \alpha_3=2$};
    \draw(2,2.15) node {$\curvearrowright$};
    \fill[black](2,2) circle (2pt);
    \draw(1.5,1.7) node{$(2,2)$};
   \end{tikzpicture}
\end{center}

It is also obvious geometrically from the direction of the normal vector to $\Pi$ that the number of solutions to the Frobenius equation is finite and hence (\ref{polyf}) is a
polynomial in the variables $\{t^2,,\ldots\,,t^{N-1} \}\,.$

The number of independent pivot terms can be reduced further; in fact to one. Let $\gamma(\tau)$ be any function with the transformation
property
\[
\gamma\left(-\frac{1}{\tau}\right) = \tau^2 \gamma(\tau) - \frac{1}{4} \tau
\]
and let
\[
f({\bf z},\tau) = \gamma(\tau) ({\bf z},{\bf z})^2 + g({\bf z},\tau)\,.
\]
Then equation (\ref{transformationlaw}) implies that
\[
g\left(\frac{{\bf z}}{\tau},-\frac{1}{\tau}\right) = \frac{1}{\tau^2} g({\bf z},\tau)\,.
\]
Thus one can obtain a refinement of Proposition \ref{basicexample}:

\begin{prop}\label{basicexample2}
The prepotential of a modular Frobenius manifold takes the form
\[
F=\frac{1}{2} u^2 \tau - \frac{1}{2} u ({\bf z},{\bf z}) + \gamma(\tau) ({\bf z},{\bf z})^2 + g({\bf z},\tau)
\]
where
\[
\gamma\left(-\frac{1}{\tau}\right) = \tau^2 \gamma(\tau) - \frac{1}{4} \tau\,,\qquad g\left(\frac{{\bf z}}{\tau},-\frac{1}{\tau}\right) = \frac{1}{\tau^2} g({\bf z},\tau)\,.
\]
Moreover, if
\[
g({\bf z},\tau) = \sum_{{\bf \alpha}\in \LL\cap \Pi} \left\{ \prod_{i=2}^{N-1} (z^i)^{\alpha_i} \right\} g_{\bf \alpha}(\tau)
\]
then
\[
g_{\bf\alpha}\left(-\frac{1}{\tau}\right) = \tau^{ (\sum \alpha_i)-2 } g_{\bf \alpha}(\tau)\,.
\]
\end{prop}

Thus modularity and quasi-homogeneity determine the modularity properties of the functions $\gamma$ and $g_{\bf\alpha}$ together
with the form of the monomial coefficients of these functions. To fix these functions one must now solve the WDVV equations
(and by construction any solution will give a Frobenius manifold).

\medskip

\textbf{Remark}. Once the Euler vector field is fixed, the solutions of the Frobenius equation \eqref{FrobeniusEquation}
determine the form of the prepotential. Imposing the WDVV equations then gives systems of over-determined non-linear ordinary
differential equations (in the variable $t^N$). It is these systems, determined by the specified $d_i$ (subject, of course, to the constraint
$d_i+d_{N+1-i}=1$) that are studied in this paper.  These equations will possess similar properties to the Chazy equation in Example \ref{exampleChazy}, their solutions
will have an $SL(2,\mathbb{C})$ symmetry. To get a modular Frobenius manifold one then looks for
solutions which have an invariance under $SL(2,\mathbb{Z}).$ In addition, if one then imposes the condition of semi-simplicity on the resulting multiplication one finds
that this then leads to very strong constraints on the possible $d_i\,.$

\medskip

Another functional class in which to look for modular manifolds could be obtained by weakening the polynomial ansatz and replacing it with a rational
ansatz (e.g rational in the variable $t^{N-1}$ and polynomial in the variables $t^1\,,\ldots\,,t^{N-2}$) - this would then include the $A_{N-2}$ and $B_{N-2}$ examples of Bertola \cite{B}\,.
These examples, while rational, are constrained by the condition that the functions $g_{\alpha}(\tau)$ are never of negative weight, i.e. $\sum \alpha_i \geq 2\,.$ Such a generalization will not
be pursued here.

\section{Solutions of the WDVV equations}\label{WDVVsection}

Proposition \ref{basicexample2} gives the forms of the prepotential of a modular Frobenius manifold: the precise form of the
modular functions that appear cannot be obtained from this analysis - this is fixed by imposing the WDVV equations of
associativity. These equations decompose into three classes which depend on the number of $\tau$-derivatives. Consider the
obstructions to associativity:
\[
\Delta[X,Y,Z]=(X\circ Y)\circ Z - X \circ (Y \circ Z)\,.
\]
Since in the case being considered we have a unity element these simplify further: if any of the vector fields
are equal to the unity field then $\Delta$ vanishes identically. Since the vector field $\partial_\tau$
is special (for example, it behaves differently to the other variables under modularity transformation), we decompose these equations further, taking the inner product with arbitrary vector
fields to obtain scalar-valued equations.

\begin{prop}\label{wdvvbreakdown} The WDVV equations for a multiplication with unity field are equivalent to
the vanishing of the following functions:
\[
\begin{array}{ccl}
\Delta^{(1)}({\bf u},{\bf v}) & = & \eta(\partial_\tau \circ \partial_\tau, {\bf u} \circ {\bf v}) - \eta(\partial_\tau \circ {\bf u}, \partial_\tau \circ {\bf v}) \,, \\ &&\\
\Delta^{(2)}({\bf u},{\bf v},{\bf w}) & = & \eta(\partial_\tau \circ {\bf u}, {\bf v} \circ {\bf w}) - \eta(\partial_\tau \circ {\bf w}, {\bf u} \circ {\bf v})\,,\\ &&\\
\Delta^{(3)}({\bf u},{\bf v},{\bf w},{\bf x}) & = & \eta({\bf u} \circ {\bf v},{\bf w}\circ {\bf x}) - \eta({\bf u} \circ {\bf x}, {\bf v} \circ {\bf w})\,
\end{array}
\]
for all ${\bf u}\,,{\bf v}\,,{\bf w}\,,{\bf x}\in  {\rm span}\{\partial_{t^i}\,i=2\,,\ldots\,,N-1\}\,.$
\end{prop}

\noindent In terms of coordinate vector fields these conditions are:
\begin{eqnarray*}
\Delta^{(1)}_{ij} &=&\eta_{ij}\, c_{\tau\tau\tau} + \eta^{pq} \left\{ c_{\tau\tau p} \,c_{ijq} - c_{\tau ip} \,c_{\tau jq}\right\}\,, \\
\Delta^{(2)}_{ijk} & = & \left\{ \eta_{jk}\, c_{\tau\tau i} - \eta_{ij} \,c_{\tau\tau k} \right\} +
\eta^{pq} \left\{ c_{\tau i p}\, c_{jkq} - c_{\tau kp}\, c_{ijq}\right\}\,, \\
\Delta^{(3)}_{ijrs} & = & \left\{ \eta_{ij} \, c_{\tau rs} + \eta_{rs} \, c_{\tau ij}-\eta_{is} \, c_{\tau rj}-\eta_{rj} \, c_{\tau is}\right\} +
\eta^{pq} \left\{ c_{ijp} \,c_{rsq} - c_{isp} \,c_{rjq}\right\}\,
\end{eqnarray*}
where $\eta_{ij}=-(\partial_i,\partial_j)\,.$

By judicious choice of variables one can obtain first order evolution equations for the unknown fields. However these systems are
overdetermined, so one has to analyze various cases and subcases very carefully.

\subsection{$N=3$ solutions}

In this simplest case (see Example \ref{exampleChazy}) there is no freedom: the Euler vector field is fixed by the fact that $d=1$ and hence

\[
E=t^1 \frac{\partial\phantom{t^1}}{\partial t^1}+ \frac{1}{2} t^2 \frac{\partial\phantom{t^2}}{\partial t^2}
\]
The Frobenius equation (\ref{FrobeniusEquation}) only has one solution and hence modularity and quasi-homogeneity
imply a polynomial prepotential of the form\footnote{In examples {\sl only} we lower indices for notational convenience, i.e. in expressions for the prepotentials in
specific examples, $t_i=t^i\,.$}
\[
F= \frac{1}{2} t_1^2 t_3 + \frac{1}{2} t_1 t_2^2 - \frac{1}{16} t_2^4 \gamma(t_3)\,.
\]
The WDVV equations then imply that $\gamma$ must satisfy the third order equation
\begin{equation}
\gamma^{'''} - 6 \gamma \gamma^{''} + 9 (\gamma^{'})^2=0\,.
\label{chazy}
\end{equation}
This is nothing more than the Chazy equation, whose modularity properties are well known (see \cite{dubrovin1}).
This equation falls into Chazy's class (see appendix) and its solutions may be written in terms of the Schwarzian
triangle function $S[\frac{1}{2},\frac{1}{3},0,t]\,.$

\subsection{$N=4$ solutions}

In this case there is a $1$-parameter family of possible Euler fields, namely
\[
E=t^1 \frac{\partial\phantom{t^1}}{\partial t^1}+ (1-\sigma) t^2 \frac{\partial\phantom{t^2}}{\partial t^2}+
\sigma t^3 \frac{\partial\phantom{t^3}}{\partial t^3}.
\]
Without loss of generality we may assume decreasing degrees and hence $\sigma\leq \frac{1}{2}\,.$ A detailed analysis of the Frobenius equation (\ref{FrobeniusEquation}) gives the
solutions summarized in Table 1.

\begin{table}
\begin{center}
\begin{tabular}{c|c}
$\sigma$& Non-pivot solutions $(\alpha_2,\alpha_3)$: \\ \hline
$\frac{1}{2}$ & $(4,0)\,,(3,1)\,,(1,3)\,,(0,4)\,;$\\
& \\
$\frac{1}{3}$ & $(3,0)\,,(1,4)\,,(0,6)\,;$\\
& \\
$\frac{1}{n}\,,(n\geq 4)$ & $(1,n+1)\,,(0,2n)\,;$\\
& \\
$\frac{2}{2n+1}\,,(n\geq 2)$ & $(0,2n+1)\,;$\\
& \\
arbitrary & none\,.
\end{tabular}
\end{center}
\vskip 5mm
\caption{Non-pivot solutions of the Frobenius equation $( \ref{FrobeniusEquation} )$ for $N=4\,.$}
\end{table}

In all cases except  $\sigma\in\{\frac{1}{2}\,,\frac{1}{3}\}$ the resulting solution of the WDVV equations
gives a nilpotent algebra - one obtains Euler-type differential equations which may be solved explicitly,
and the solution checked to see if the family of algebras it defines is semi-simple or not. The two remaining cases
require a more detailed analysis.

\subsubsection{$N=4\,,\sigma=\frac{1}{2}\,$}

In this case we take
\[
F=\frac{1}{2} t_1^2 t_4 +t_1 t_2 t_3 - \frac{1}{4} (t_2 t_3)^2 \gamma(t_4) +
\left\{ t_3^4 \,g_1(t_4)  + t_2 t_3^3 \,g_2(t_4) + t_2^3 t_3 \,g_3(t_4) + t_2^4 \,g_4(t_4)  \right\}\,.
\]
The analysis of the WDVV equations splits into two subcases (there are other subcases which appear in the analysis, but these yield non-semi-simple solutions):

\medskip

\noindent{\underline{Case I:}} $g_2=0\,,g_3=0$

\medskip

\noindent One finds that $g_4 = \mu g_1$ where $\mu$ is a constant and the evolution of $g_1$ and $\gamma$ are given by the equations:

\begin{equation}
\begin{array}{rcl}
\gamma^\prime & = & \displaystyle{\frac{1}{2}} \gamma^2 - 288 \mu g_1^2\,,\\
g_1^{\prime\prime} & = & 3 \gamma \, g_1^\prime - 3 g_1 \gamma^\prime\,.
\end{array}
\label{CaseI}
\end{equation}
\medskip

\noindent On eliminating $g_1$ one obtains a third-order scalar equation

\begin{equation}
{\dddot{\gamma}}= \frac{1}{2} \frac{(\ddot{\gamma}- 2 \gamma \dot{\gamma})^2}{\dot{\gamma} - \gamma^2} +8 \gamma \ddot{\gamma} -10  \dot{\gamma}^2\,.
\label{caseIueqn}
\end{equation}
Here the independent variable has been rescaled, $t=\frac{1}{2} t_4\,,$ and hence $\gamma'=\frac{1}{2} \dot{\gamma}\,$ etc..
This falls within Bureau's class (see appendix) and its solutions are given in terms of the
Schwarzian triangle function $y(t)=S[\frac{1}{2},\frac{1}{4},0,t]\,$, namely
\begin{eqnarray*}
\gamma(t) & = & \frac{1}{2} \left\{ \frac{\ddot{y}}{\dot{y}} - \left(\frac{1/2}{y} + \frac{3/4}{y-1}\right) \dot{y}\right\}\,,\\
& = & \frac{1}{2} \frac{d\phantom{t}}{dt} \log \left\{ \frac{\dot{y}}{y^\frac{1}{2} (y-1)^\frac{3}{4}}\right\}\,,\\
g_1(t) & = & \frac{1}{192 \mu^\frac{1}{2}} \, \frac{1}{y^\frac{1}{2} (y-1)^\frac{1}{2}} \dot{y}\,.
\end{eqnarray*}

An alternative way to solve (\ref{caseIueqn}) (following Satake \cite{SatakeD4} and section \ref{foldingsection}) is to express the solutions in terms of solutions to the Halphen system
\begin{eqnarray*}
\dot{\omega}_1 & = & - \omega_2 \omega_3 + \omega_1(\omega_2+\omega_3)\,,\\
\dot{\omega}_2 & = & - \omega_3 \omega_1 + \omega_2(\omega_3+\omega_1)\,,\\
\dot{\omega}_3 & = & - \omega_1 \omega_2 + \omega_3(\omega_1+\omega_2)\,.\\
\end{eqnarray*}
One may show that the solution is:
\begin{eqnarray*}
\gamma(t) & = & \displaystyle{\frac{1}{4}(\omega_1 + 2 \omega_2 + \omega_3)}\,,\\
g_1(t) & = & \displaystyle{\frac{1}{96 \mu^\frac{1}{2}}(\omega_1 - 2 \omega_2 + \omega_3)}\,.
\end{eqnarray*}
The required modularity properties of $\gamma$ and $g_1$ then follow automatically from the known modularity properties of the solution to the Halphen system.

\medskip

\noindent{\underline{Case II:}} $g_2\neq0\,,g_3\neq0$

\medskip

\noindent One finds that $g_2=\frac{4}{K} g_4\,,g_3= K g_1$ where $K$ is a constant and the evolution of $g_1\,,g_4$ and $\gamma$ are given by the equations:

\begin{equation}
\begin{array}{rcl}
\gamma^\prime & = & \displaystyle{\frac{1}{2}} \gamma^2 - 288 g_1 g_4\,,\\
g_1^\prime & = & \gamma \, g_1 + 24 K g_4^2\,,\\
g_4^\prime & = & \gamma \, g_4 + 24 K^{-1} g_1^2\,.
\end{array}
\label{CaseII}
\end{equation}

\noindent Eliminating $g_1$ and $g_4$ yields the Chazy equation \eqref{chazy}.
From the expressions for $\gamma^\prime$ and $\gamma^{\prime\prime}$ one can easily obtain two algebraic relations connecting $g_1$ and $g_4$ to the $\gamma$.
Hence one can obtain,
by solving these algebraic equations, the general solution in this subcase. Since it is well known that the Halphen system is equivalent to the Chazy equation, one may also
express the solution in terms of the $\omega_i\,.$
One can also show, without first having to solve the equations, that the solution is semi-simple.

\subsubsection{$N=4\,,\sigma=\frac{1}{3}\,$}

In this case we take
\[
F=\frac{1}{2} t_1^2 t_4 +t_1 t_2 t_3 - \frac{1}{4} (t_2 t_3)^2 \gamma(t_4) +
\left\{ t_3^6 \,g_4(t_4)  + t_2 t_3^4\, g_3(t_4) + t_2^3\, g_1(t_4)  \right\}\,.
\]
The analysis of the WDVV equations gives the single system (there are other subcases which appear in the analysis, but these yield non-semi-simple solutions):

\begin{equation}
\begin{array}{rcl}
\gamma^\prime & = & \displaystyle{\frac{1}{2}} \gamma^2 - 72 K g_1^4\,,\\
g_1^{\prime\prime} & = & 2 \gamma \, g_1^\prime -  g_1 \gamma^\prime\,
\end{array}
\label{CaseIII}
\end{equation}
where $K$ is a constant and
\[
g_3 = K g_1^3\,,\quad g_4 = \frac{K g_1}{30} \left[ g_1^\prime - \frac{1}{2} g_1 \gamma\right]\,.
\]

\noindent On eliminating $g_1$ one obtains a third-order scalar equation
\[
{\dddot{\gamma}}= \frac{3}{4} \frac{(\ddot{\gamma}- 2 \gamma \dot{\gamma})^2}{\dot{\gamma} - \gamma^2} +6 \gamma \ddot{\gamma} -6  \dot{\gamma}^2\,.
\]
Here the independent variable has been rescaled, $t=\frac{1}{2} t_4\,,$ and hence $\gamma'=\frac{1}{2} \dot{\gamma}\,$ etc..
This falls within Bureau's class (see appendix) and its solutions are given in term of the
Schwarzian triangle function $y(t)=S[\frac{1}{2},\frac{1}{6},0,t]\,.$

\begin{eqnarray*}
u(t) & = & \frac{1}{2} \left\{ \frac{\ddot{y}}{\dot{y}} - \left(\frac{1/2}{y} + \frac{5/6}{y-1}\right) \dot{y}\right\}\,,\\
& = & \frac{1}{2} \frac{d\phantom{t}}{dt} \log \left\{ \frac{\dot{y}}{y^\frac{1}{2} (y-1)^\frac{5}{6}}\right\}\,.
\end{eqnarray*}
Given this solution one may easily find the remaining functions: they all take the schematic form
\[
g_i(t) = \frac{c_i}{y^{a_i}(y-1)^{b_i}} \left(\dot{y}\right)^\frac{i}{2}\,,\qquad i=1\,,3\,,4
\]
for various constants $a_i\,,b_i\,,c_i\,.$

\subsection{$N=5$ solutions}

In this case there is again a $1$-parameter family of Euler vector fields,
\[
E=t^1 \frac{\partial\phantom{t^1}}{\partial t^1}+ (1-\sigma) t^2 \frac{\partial\phantom{t^2}}{\partial t^2}+
\frac{1}{2} t^3 \frac{\partial\phantom{t^3}}{\partial t^3}+\sigma t^4 \frac{\partial\phantom{t^4}}{\partial t^4}.
\]
In this case we have 3 pivot points which correspond to the monomials $\{ (t_2t_4)^2, t_2(t_3)^2t_4, (t_3)^4 \}$. Again, without loss of generality we may assume decreasing
degrees and hence $\sigma\leq \frac{1}{2}\,.$ Analyzing the Frobenius equation \eqref{FrobeniusEquation}, we find the non-pivot terms summarized in table 2.
\begin{table}
\begin{center}
\begin{tabular}{c|c}
$\sigma$& Non-pivot solutions $(\alpha_2,\alpha_3,\alpha_4)$: \\ \hline \\
$\frac{1}{2}$ & $(4,0,0)\,,(3,0,1)\,,(1,0,3)\,,(0,0,4)\,$\\ & $(3,1,0)\,,(2,1,1)\,,(1,1,2)\,,(0,1,3)\,$ \\ & $(2,2,0)\,,(0,2,2)\,$ \\ & $(1,3,0)\,,(0,3,1);$\\
& \\
$\frac{1}{3}$ & $(3,0,0)\,,(1,0,4)\,,(0,0,6)\,(0,2,3);$\\
& \\
$\frac{1}{4}$ & $(1,0,5)\,,(0,0,8)\,$ \\ & $(2,1,0)\,,(1,1,3)\,,(0,1,6)\,$ \\& $(0,2,4)\,,(0,3,2);$\\
& \\
$\frac{1}{n}\,,(n\geq 5\,,{\rm and ~odd})$ & $(1,0,n+1)\,,(0,0,2n)\,(0,2,n);$\\
& \\
$\frac{2}{2n+1}\,,(n\geq 2)$ & $(0,0,2n+1)\,;$\\
& \\
$\frac{1}{2n}\,, (n\geq 2)$ & $(1,0,2n+1)\,,(0,0,4n)\,, (1,1,n+1)\,, (0,1,3n)\,$ \\  & $(0,2,2n)\,, (0,3,n)\:$\\
& \\
$\frac{3}{2n}\,, (n\geq 4)$ & $(0,1,n)\:$\\
& \\
arbitrary & none\,.
\end{tabular}
\end{center}
\vskip 5mm
\caption{Non-pivot solutions of the Frobenius equation $( \ref{FrobeniusEquation} )$ for $N=5\,.$}
\end{table}

Once our polynomial ansatz is made, WDVV reduces to systems of ordinary differential equations for the pivot and non-pivot functions. In all cases, except for $\sigma = 1/2$, these equations
typically fall into one of two types:
\begin{itemize}
\item[1.]{second order Euler-type differential equations that may be integrated explicitly to yield a non-semi-simple solution to WDVV.}
\item[2.]{first order equations of the form
\begin{equation}\label{zerocurv}
\gamma' - \frac{1}{2}\gamma^2 =0,
\end{equation}
where $\gamma$ is a pivot term.
}
\end{itemize}
\noindent In such cases the resulting Frobenius manifold is non-semi-simple and therefore $\sigma=1/2$ is the only value for which one may hope to find a semi-simple
modular Frobenius manifold. Note that this is the most symmetric case - the prepotential contains 15 unknown functions (3 pivot and 12 non-pivot) and homogeneity means
that the overdetermined system to be solved is far more complicated than in other examples. We note, though, two solutions which have appeared in the literature.

\begin{example}\label{MPsolution}
The following example (rescaled slightly) was found in \cite{MartiniPost}.
\begin{eqnarray*}
F & = & \displaystyle{\frac{1}{2} t_1^2 t_5 + t_1 \left(t_2 t_4 + \frac{1}{2} t_3^2 \right) - \frac{1}{4} \left(t_2 t_4 + \frac{1}{2} t_3^2 \right)^2 \gamma(t_5)}\\
& & \displaystyle{ +\frac{1}{24} \left( t_2^4 - 2 t_2 t_3^3 - 2 t_2 t_4^3 + 3 t_3^2 t_4^2 \right) h_1(t_5) }\displaystyle{ +\frac{1}{24} \left( t_4^4 - 2 t_4 t_3^3 - 2 t_2^3 t_4 + 3 t_2^2 t_3^2
\right)
h_2(t_5) }\\
\end{eqnarray*}
where
\begin{eqnarray*}
\gamma^\prime & = & \frac{1}{2} \gamma^2-\frac{1}{2} h_1 h_2\,,\\
h_1^\prime & = & \gamma h_1 + h_2^2\,,\\
h_2^\prime & = & \gamma h_2 + h_1^2\,.
\end{eqnarray*}
This system reduces to the Chazy equation for the function $\gamma\,.$

\end{example}

\begin{example}\label{SatakeFolded}
The following example may be obtained by \lq folding\rq~the $D_4^{(1,1)}$ solution found by Satake \cite{SatakeD4} (see section \ref{foldingsection}),
reducing a six dimensional solution to a five dimensional solution.
\begin{eqnarray*}
F & = & \displaystyle{\frac{1}{2} t_1^2 t_5 + \frac{1}{2} t_1 \left( t_2^2 + t_3^2 + t_4^2 \right) - \frac{1}{16} \left( t_2^2 + t_3^2 + t_4^2 \right)^2 \gamma(t_5)}\\
& & + \displaystyle{\left( t_2^4 + t_3^4 - 2 t_3^2 t_4^2 - 2 t_2^2 t_3^2 - 2 t_2^2 t_4^2\right) g_1(t_5) + \frac{1}{2} (t_2 t_3 t_4^2) g_2(t_5)\,,}\\
\end{eqnarray*}
where
\begin{eqnarray*}
\gamma^\prime & = & \displaystyle{ \frac{1}{2} \gamma^2 - \frac{2}{3} g_2^2 -128 g_1^2 }\,,\\
g_1^\prime & = & \displaystyle{ \gamma g_1 - \frac{1}{12} g_2^2 + 16 g_1^2 } \,,\\
g_2^\prime & = & \displaystyle{ \gamma g_2 - 32 g_1 g_2}\,.
\end{eqnarray*}
This system too reduces to the Chazy equation for the function $\gamma\,.$
\end{example}
Note that the two systems in these examples are inequivalent - the corresponding quadratic
algebras are non-isomorphic (see section \ref{modulardynamialsection} and \cite{Ohyama}).

\section{Modular almost-dual solutions}

In \cite{dubrovin2} Dubrovin introduced the notion of an almost-dual Frobenius manifold. Consider the vector field $E^{-1}$ defined by the condition
\[
E^{-1} \circ E = e\,.
\]
This is defined on $M^{\star}= M\backslash \Sigma\,,$ where $\Sigma$ is the discriminant submanifold
on which $E^{-1}$ is undefined. With this field one may define a new \lq dual\rq~multiplication
$\star: TM^\star \times TM^\star \rightarrow TM^\star$ by
\[
X \star Y = E^{-1} \circ X \circ Y\,, \qquad\qquad \forall\, X\,,Y \in TM^\star\,.
\]
This new multiplication is clearly commutative and associative, with the Euler vector field being the
unity field for the new multiplication.

Furthermore, this new multiplication is compatible with the intersection form $g$ on the Frobenius manifold,
i.e.
\[
g(X\star Y, Z) = g(X,Y\star Z)\,, \qquad\qquad \forall\, X\,,Y\,,Z \in TM^\star\,.
\]
Here $g$ is defined by the equation
\[
g(X,Y)=\eta(X\circ Y, E^{-1})\,, \qquad\qquad \forall\, X\,,Y \in TM^\star
\]
(and hence is well-defined on $M^\star\,$).
The intersection form has the important property that it is flat, and hence there exists a distinguished
coordinate system $\{ {\bf p}\}$ in which the components of the intersection form are constant. It turns
out that there exists a dual prepotential $F^\star$ such that its third derivatives give the structure
functions $c^{\star}_{ijk}$ for the dual multiplication. More precisely \cite{dubrovin2}:

\begin{theorem} Given a Frobenius manifold $M$, there exists a function $F^\star$ defined on $M^\star$
such that:
\begin{eqnarray*}
c^{\star}_{ijk} & = &
g\left( \frac{\partial~}{\partial p^i}\star \frac{\partial~}{\partial p^j}\,, \frac{\partial~}{\partial p^k}
\right) \,,\\
& = &\frac{\partial^3 F^\star}{\partial p^i \partial p^j \partial p^k}\,.
\end{eqnarray*}
Moreover, the pair $(F^\star,g)$ satisfies the WDVV-equations in the flat coordinates $\{ {\bf p} \}$ of the metric $g\,.$
\end{theorem}
Thus starting with a Frobenius manifold one may construct an \lq almost-dual\rq~Frobenius manifold (for full details see  \cite{dubrovin2})
with prepotential $F^\star\,:$
\[
F \stackrel{\star}{\longrightarrow} F^\star
\]
Combining this with the symmetry $I$ leads to the following picture
\[
\begin{array}{ccc}
F&\overset{I}{\longrightarrow}&
{\hat F}\\
\downarrow & & \downarrow \\
F^\star
&&
{\hat F}^\star
\end{array}
\]
In \cite{MS} the induced involutive symmetry
\[
F^\star \overset{I^\star}{\longrightarrow} {\hat F}^\star\,.
\]
was constructed. It turns out that the action $I^\star$ is different for modular and non-modular Frobenius manifolds.

\begin{theorem}
Let $F$ define a Frobenius manifold and let $\hat{F}$ denote the induced manifold under the action of the symmetry $I\,.$ Let $F^\star$ and
$\hat{F^\star}$ denote the corresponding almost dual structures. The $I^\star$, the induced symmetry act as:

\begin{itemize}

\item{Case I: $d \neq 1\,:$}

\begin{eqnarray*}
{\hat p}^i & = & \frac{p^i}{t_1} \,, \qquad i=1\,,\ldots\,, N\,,\\
{\hat{g}}_{ab} & = & g_{ab} \,, \\
{\hat{F^\star}} & = & \frac{F^\star}{t_1^2}
\end{eqnarray*}
where $t_1=\frac{1}{2} g_{ab} p^a p^b\,.$

\medskip

\item{Case II: $d=1\,:$}

\begin{eqnarray*}
{\hat p}^1 & = & \frac{1}{2}\frac{p_{\sigma}p^{\sigma}}{t_1}\,,\\
{\hat p}^i & = & \frac{p^i}{t_1} \,, \qquad i=2\,,\ldots\,, N-1\,,\\
{\hat p}^N & = & -\frac{1}{t_1}\,,\\
{\hat{g}}_{ab} & = & g_{ab} \,, \\
{\hat{F^\star}}(\hat{p}) & = & (\hat{p}^N)^2 F \left( \frac{1}{2}\frac{\hat{p}_{\sigma}\hat{p}^{\sigma}}{\hat{p}^N}, -\frac{\hat{p}^{2}}{\hat{p}^N}, ... ,
-\frac{\hat{p}^{N-1}}{\hat{p}^N},   -\frac{1}{\hat{p}^N} \right) + \frac{1}{2}\hat{p}^1\hat{p}_{\sigma}\hat{p}^{\sigma}\,.
\end{eqnarray*}

\end{itemize}

\end{theorem}

Thus if one starts with a modular Frobenius manifold the almost-dual prepotential will also satisfy Proposition \ref{basicexample}.

In all known examples (essentially the Jacobi group orbit spaces for $A_{N-2}$ and $B_{N-2}$ for arbitrary $N\,$) the (modular) almost dual
prepotential takes the form
\begin{equation}
F^\star = \frac{1}{2} u^2 \tau - \frac{1}{2} u ({\bf z},{\bf z}) + \sum_{\alpha \in \mathfrak{U}} h_\alpha f\left( (\alpha,{\bf z}) ,\tau\right)
\label{almostdualprepotentialansatz}
\end{equation}
where $f(z,\tau)$ is constructed from the elliptic trilogarithm (see \cite{iabs2}) and has the expansion
\[
\begin{array}{rcl}
\label{fseries1}
f(z,\tau) & = & \displaystyle{-\frac{1}{(2 \pi i)} \left\{ \frac{1}{2} z^2 \log z + z^2 \log \eta(\tau) \right\}} \\
&&\\
&&\displaystyle{+\frac{1}{(2\pi i)^3} \sum_{n=1}^\infty \frac{ (-1)^n E_{2n}(\tau) B_{2n} }{(2n+2)! (2n)}(2 \pi z)^{2n+2}}\,,\qquad {\rm(mod~quadratic~terms)}
\end{array}
\]
where $E_{2n}$ are Eisenstein series - and hence are modular forms not functions - and it follows from this that $f$ has the required modularity property
\[
f\left(\frac{z}{\tau}, -\frac{1}{\tau}\right) = \frac{1}{\tau^2} f(z,\tau) - \frac{1}{\tau^3} \frac{z^4}{4!}\,,\qquad {\rm(mod~quadratic~terms)}\,.
\]
The set of vectors $\mathfrak{U}$ is constructed via a Landau-Ginzburg/Hurwitz space construction \cite{RS1}. These satisfy the conditions
\begin{eqnarray*}
\sum_{\alpha\in\mathfrak{U}} h_\alpha (\alpha,{\bf z})^2 & = & 0 \,,\\
\sum_{\alpha\in\mathfrak{U}} h_\alpha (\alpha,{\bf z})^4 & = & 3 ({\bf z},{\bf z})^2 \,,\\
\end{eqnarray*}
and it then follows that the basic modularity property given in Proposition \ref{basicexample} is satisfied. Using these conditions the almost dual prepotentials
may be written in a slightly modified version of Proposition \ref{basicexample2}
\begin{eqnarray*}
F^\star & = &\displaystyle{ \frac{1}{2} u^2 \tau - \frac{1}{2} u ({\bf z},{\bf z}) - \frac{2 \pi i}{96} E_2(\tau)\, ({\bf z},{\bf z})^2 - \frac{1}{2(2\pi i)}
\sum_{\alpha\in\mathfrak{U}} h_\alpha (\alpha,{\bf z})^2 \log(\alpha.{\bf z})} \\
& & \displaystyle{+ \sum_{n=2}^\infty c_n E_{2n}(\tau) \left\{ \sum_{\alpha\in\mathfrak{U}} h_\alpha (\alpha.{\bf z})^{(2n+2)}\right\}}\,.
\end{eqnarray*}
Note the single pivot-term, proportional to $E_2(\tau)$\,.

In \cite{iabs2} sufficient conditions were given on an arbitrary set of vectors $\mathfrak{U}$ to ensure that a function of the form (\ref{almostdualprepotentialansatz}) is a
solution of the WDVV equations. However, determining
further conditions which would ensure whether or not the solution is also the almost-dual potential of a modular Frobenius manifold (or even a polynomial modular Frobenius manifold)
is a considerably harder problem. Another open problem is the construction of the almost-dual manifolds corresponding to the low-dimensional polynomial modular Frobenius manifolds
constructed in the previous section. Such a construction should be deeply connected to the theory of Jacobi forms \cite{B,Satake}.

\section{Foldings and codimension 2-toroidal Lie algebras}\label{foldingsection}

As noted in the introduction, a number of polynomial modular Frobenius manifolds have appeared in the literature before. In fact, some of the earliest
solutions of the WDVV equations fall into this class and were written down before the formal definition of a Frobenius manifold.
The space of versal unfoldings of the simple elliptic singularities, denoted by different
authors as ${\widetilde{E}}_{6,7,8}\,,E^{(1,1)}_{6,7,8}\,$ or $P_8\,,X_9\,,J_{10}\,,$
carries the natural structure of a Frobenius manifold. The prepotential and Euler vector field for the $\widetilde{E}_6$ singularity
were derived by Verlinde and Warner \cite{VW}:
\begin{eqnarray*}
F&=&\frac{1}{2} t_1^2 t_8 + t_1 (t_2 t_7 + t_3 t_6 + t_4 t_5) + t_2 t_3 t_4 f_0(t_8) + \frac{1}{6}(t_2^3 + t_3^3 + t_4^3) f_1(t_8) \\
&&+(t_2 t_3 t_6 t_7 +t_2 t_4 t_5 t_7 + t_3 t_4 t_5 t_6)f_2(t_8) + \frac{1}{2} (t_2^2t_5t_6 +t_3^2t_5t_7 + t_4^2t_6t_7)f_3(t_8)\\
&&+(t_2 t_3 t_5^2 + t_2 t_4 t_6^2 + t_3 t_4 t_7^2)f_4(t_8) + \frac{1}{4}(t_2^2 t_7^2 +  t_3^2 t_6^2 + t_4^2 t_5^2)f_5(t_8)\\
&&+\frac{1}{6} [ t_2 t_7 (t_5^3 + t_6^3) +t_3 t_6 (t_5^3 + t_7^3)+t_4 t_5 (t_6^3 + t_7^3)] f_6(t_8)\\
&&+\frac{1}{2} (t_2 t_5 t_6 t_7^2+t_3 t_5 t_6^2 t_7+t_3 t_5^2 t_6 t_7)f_7(t_8) +
\frac{1}{4} ( t_2 t_5^2 t_6^2  +t_3 t_5^2 t_7^2 +t_4 t_6^2 t_7^2 ) f_8(t_8)\\
&&+\frac{1}{24} ( t_2 t_7^4 + t_3 t_6^4 + t_4 t_5^4)f_9(t_8) + \frac{1}{36}( t_5^3 t_6^3 +t_5^3 t_7^3+t_6^3 t_7^3)f_{10}(t_8) \\
&&+\frac{1}{24} ( t_5 t_6 t_7^4 + t_5 t_6^4 t_7 + t_5^4 t_6 t_7)f_{11}(t_8) + \frac{1}{8} t_5^2 t_6^2 t_7^2 f_{12}(t_8) + \frac{1}{720} ( t_5^6+t_6^6+t_7^6) f_{13}(t_8)\,,
\end{eqnarray*}
and
\[
E=t^1 \frac{\partial\phantom{t^1}}{\partial t^1}+
\frac{2}{3} \left(t^2 \frac{\partial\phantom{t^2}}{\partial t^2}+
 t^3 \frac{\partial\phantom{t^3}}{\partial t^3}+
 t^4 \frac{\partial\phantom{t^4}}{\partial t^4}\right)+
\frac{1}{3}\left( t^5 \frac{\partial\phantom{t^5}}{\partial t^5}+
 t^6 \frac{\partial\phantom{t^6}}{\partial t^6}+
 t^7 \frac{\partial\phantom{t^7}}{\partial t^7}\right)\,.
\]
The functions $f_i$ may all be expressed in terms of the Schwarzian triangle function
$S[\frac{1}{2},\frac{1}{2},\frac{1}{6},t]\,.$ Note that $f_2$ and $f_5$ are the pivot-terms.
The prepotential for $\widetilde{E}_7$ was derived by \cite{KTS}. The prepotential for $\widetilde{E}_8$
has not been written down - one would expect it to
have over 200 terms. In principle this is just a computational exercise - one can find the
multiplication in terms of the natural deformation variables in the Jacobi ring and then use the flat
coordinates derived by Noumi and Yamada to derive the prepotential \cite{NY}. It is clear, though, that this
prepotential will fall into the class of polynomial modular Frobenius manifolds considered here. It is interesting to note that
the invariance of such solutions under an inversion symmetry was known already by Verlinde and Warner \cite{VW}, but the geometric origins of this
symmetry were not fully understood.

The connection between the underlying singularity theory and the Frobenius manifold is subtle in these examples: one has
a one parameter family of Frobenius manifolds associated with each elliptic singularity, and to construct a global theory one has to understand object such as
the primitive form as a function of this parameter (the calculations in \cite{VW} and \cite{NY} were performed at one point in this family).
For details of this global construction see the recent preprint of Milanov and Ruan \cite{MR}. From the point of view of this paper it is
hard see this one parameter family; the modular symmetry condition in Proposition \ref{basicexample2} may impose too strong conditions to see this
family. One possibility might be to use the weaker version of modular symmetry used in Example \ref{exampleChazy}.

Another polynomial modular Frobenius manifold was derived explicitly by Satake \cite{SatakeD4} and denoted $D_4^{(1,1)}\,.$ Here the
prepotential and Euler field are given by
\begin{eqnarray*}
F&=&\frac{1}{2} t_1^2 t_6 + \frac{1}{2} t_1 \left( t_2^2+t_3^2+t_4^2+t_5^2 \right) +
\frac{1}{2} \left( t_2^2+t_3^2+t_4^2+t_5^2 \right)^2\gamma(t_6)\\
&& + \left[ t_2^4+t_3^4+t_4^4+t_5^4 - 2(t_2^2 t_3^2 + t_2^2 t_4^2 + t_2^2 t_5^2 + t_3^2 t_4^2 + t_3^2 t_5^2 + t_4^2 t_5^2) \right]\, g_1(t_6)
+\left(t_2 t_3 t_4 t_5\,\right) g_2(t_6)
\end{eqnarray*}
and
\[
E=t^1 \frac{\partial\phantom{t^1}}{\partial t^1}+
\frac{1}{2} \left(t^2 \frac{\partial\phantom{t^2}}{\partial t^2}+
 t^3 \frac{\partial\phantom{t^3}}{\partial t^3}+
 t^4 \frac{\partial\phantom{t^4}}{\partial t^4}+
 t^5 \frac{\partial\phantom{t^5}}{\partial t^5}\right)
\]

It was shown by Zuber \cite{Zuber}, in the case of polynomial Frobenius manifolds, that the non-simply-laced examples
(namely $B_n\,,F_4\,,G_2\,,H_{3,4}\,,I_2(n))$ may be obtained by restricting the simply laced examples (namely
$A_n\,,D_n\,,E_{6,7,8}$) to certain hyperplanes $\Sigma=\{ t^i=0\,, i\notin I\}$ for some subset $I \subset \{1\,,\ldots\,,N\}\,.$
In order to obtain a subalgebra one requires the condition
\[
\left.c_{ij}^{\phantom{ij}k}\right|_\Sigma = 0\,,\qquad \forall i\,,j \in I\,, \forall k \notin I\,.
\]
This process may be understood in terms of automorphisms of the
corresponding Coxeter diagrams and their foldings (see, for example, Figure \ref{E6F4}). In such diagrams the arrow should be read as \lq folds to\rq\,.

\begin{figure}[htb]
\begin{center}
\leavevmode
\includegraphics[angle=270,scale=0.18]{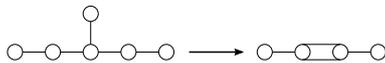}
\end{center}
\caption{The folding of $E_6$ to $F_4\,.$}
\label{E6F4}
\end{figure}

Similarly the solutions found by Verlinde and Warner \cite{VW} and Satake \cite{SatakeD4} given above also admit such foldings,
and the resulting solutions remain in the class of polynomial modular manifolds.

The $D_4^{(1,1)}$ submanifolds are\footnote{Here the subscripts on $\Sigma$ denote the codimension of the submanifold (i.e. the number of constraints on the
original, ambient, manifold). The empty set denotes the absence of any constraints and so corresponds to the original manifold. As more foldings take place the number of
constraints increase, and the dimension of the submanifold decreases.}:
\[
\begin{array}{ccccc}
&&\Sigma_{2,a}\cong \{ t_2=t_3=t_4\}&&\\
& \nearrow && \searrow & \\
\emptyset \rightarrow \Sigma_1 \cong \{t_4=t_5\} &&&&  \Sigma_3 \cong \{t_2=t_3=t_4=t_5\} \\
& \searrow && \nearrow & \\
&&\Sigma_{2,b}\cong \{ t_2=t_3\,,t_4= t_5 \}&&\\
\end{array}
\]
and the $E_6^{(1,1)}$ submanifolds are:
\[
\emptyset \rightarrow \Sigma_1 \cong \{t_2=t_3\,,t_6= t_7 \} \rightarrow \{ t_2=t_3=t_4\,,t_5=t_6=t_7 \}\,.
\]
The corresponding foldings are shown Figures \ref{D4fold} and \ref{E6fold}. Note that the $E_6^{(1,1)}$ manifold does not admit a folding to a five dimensional manifold, in
agreement with the results of section \ref{WDVVsection}. These foldings may also be understood in terms of
boundary unimodal singularities \cite{ArnoldBook,Arnold}, and in terms of the theory of 2-toroidal Lie algebras \cite{saito3,SatakeD4}. More details on the meaning of these
more complicated diagrams may be found in \cite{saito3}. However, there are subtleties in the associated invariant theory \cite{STletter}.

\begin{figure}[htb]\label{D4fold}
\begin{center}
\leavevmode
\includegraphics[angle=270,scale=0.3]{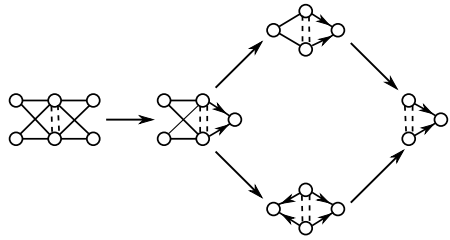}
\end{center}
\caption{The folding of $D_4^{(1,1)}\,.$}
\label{D4fold}
\end{figure}

\begin{figure}[htb]\label{E6fold}
\begin{center}
\leavevmode
\includegraphics[angle=270,scale=0.3]{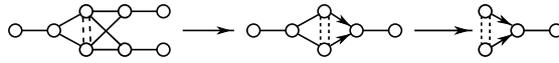}
\end{center}
\caption{The folding of $E_6^{(1,1)}\,.$}
\label{E6fold}
\end{figure}

\medskip

The folded prepotentials and the prepotentials found in section \ref{WDVVsection} do not, at first sight, agree: the
underlying triangle functions and the forms of the functions that appear are different. However, the two
different triangle functions are related by certain nonlinear transformations which have their origin in the
nonlinear Goursat identities for hypergeometric functions \cite{McKayHarnad}. For example, the relation
\[
(1-z)^{-\frac{1}{6}} {}_2F_1\left(\frac{1}{6}\,,\frac{1}{6}\,, \frac{5}{6}\,; \frac{z^2}{4(z-1)}\right) = {}_2F_1\left(
\frac{1}{3}\,,\frac{1}{3}\,,\frac{2}{3}\,;z\right)
\]
leads to the relation $\tilde{x} = \frac{x^2}{4(x-1)}$ between the Schwarzian functions $x=S[\frac{1}{3}\,,0\,,0\,,(-4)^\frac{1}{6} t]$
and $\tilde{x}=S[\frac{1}{2}\,,\frac{1}{6}\,,0\,,t]\,.$

\section{Modular dynamical systems}\label{modulardynamialsection}

The form, and modular invariance, of the equations derived in section \ref{WDVVsection} suggests that one should develop a theory of
modular dynamical systems. Suppose that one has a set of quasi-modular functions (we assume no analytic properties such as
being holomorphic in the upper-half plane) $\gamma(\tau)$ and $g_\alpha(\tau)\,, \alpha \in \cW$ for some indexing
set $\cW$ such that\footnote{The value of the non-zero constant $a$ is not fixed here - this is to facilitate comparison with the work of various sets of authors who use
different values in their work.}:
\begin{eqnarray*}
\gamma\left(-\frac{1}{\tau}\right) & = & \tau^2 \gamma(\tau) + a \tau\,,\qquad\quad a \neq 0\,{\rm ~constant\,,}\\
g_n\left(-\frac{1}{\tau}\right) & = & \tau^n g_n(\tau)\,.
\end{eqnarray*}
Such a $g_n$ is said to have weight $n\,.$ Define next a Rankin-type derivative
\begin{eqnarray*}
D(g_n) & = & \displaystyle{\frac{d g_n}{d\tau} - \frac{n\,\gamma}{a} g_n}\,,\\
D(\gamma) & = & \displaystyle{\frac{d\gamma}{d\tau} - \frac{1}{a} \gamma^2}\,.
\end{eqnarray*}
It is easy to check that $D(\gamma)$ has weight $4$ and $D(g_n)$ has weight $n+2\,$ and that $D(g_n g_m)=g_n Dg_m + g_m Dg_n\,.$
\begin{Def}
Let $\cW$ be some (finite) indexing set. A modular dynamical system takes the form
\begin{eqnarray*}
D(\gamma) & = & q({\bf g})\,,\\
D(g_\alpha) & = & p_\alpha ({\bf g})\,,\qquad \alpha \in \cW\,.
\end{eqnarray*}
where the polynomials $q$ and $p_\alpha$ are constrained so that the resulting system is invariant under the transformation
$\tau\rightarrow - \frac{1}{\tau}\,.$
\end{Def}

The interplay between the weights and the polynomials places very strong restriction on the form of these equations.
Note that if the polynomials depended on $\gamma$ as well, then the modularity property would be lost.
Thus $\gamma$ only enters through its appearance in the Rankin derivatives.

\begin{example}
Consider the 3-component system with $\cW=\{n\,,m\}\,.$ Clearly if $m\geq 5$ and $n\geq 5$ then $D\gamma=0$ and one
obtains a degenerate system (we say an $n$-component system is degenerate if the differential equation for $\gamma$
is of order less than $n$) and a similar analysis shows that the systems
$\{4,m\geq 7\}\,, \{3,m\geq 6\}\,,\{2,m\geq 5\}\,,\{1,m\geq 5\}$ are also degenerate. Thus there can only be a finite number
of non-degenerate systems. Certain other cases are degenerate, and the final list of non-degenerate cases is:
\[
\cW = \{6,4\}\,,\{4,3\}\,,\{4,2\}\,,\{4,1\}\,,\{3,1\}\,,\{2,2\}\,,\{2,1\}\,,\{1,1\}\,.
\]

\medskip

\noindent Consider now the system $\cW=\{6\,,4\}\,.$ Then the modular dynamical system must take the form:
\begin{eqnarray*}
D(\gamma)  &  = & a g_4\,,\\
D(g_4) & = & b g_6\,,\\
D(g_6) & = & c g_4^2
\end{eqnarray*}
for constants $a\,,b\,,c\,.$ A famous example in this class is the Ramanujan differential equations for the Eisenstein
series $\{E_2\,,E_4\,,E_6\}\,.$ Eliminating $g_4$ and $g_6$ gives the third order system
\[
D^3 \gamma + k (D\gamma)^2 =0
\]
where $k=-\frac{bc}{a}$ (if $a=0$ the system is degenerate). This system again falls into the class of generalized Chazy equations
considered in the appendix and its solutions can be expressed in terms of Schwarzian triangle functions.

\end{example}

The differential equations obtained in section \ref{WDVVsection} may be written as modular dynamical systems:

\bigskip

\noindent{\underline{$N=4\,,\sigma=\frac{1}{2}$ Case I:} The system (\ref{CaseI}) may be written as:
\[
D(h) = -4 g^3\,,\qquad D(g) = h \,, \qquad D(\gamma) = g^2\,.
\]
Here the weights of $g (=g_1)$ and $h$ are 2 and 4 respectively.

\bigskip

\noindent{\underline{$N=4\,,\sigma=\frac{1}{2}$ Case II:} The system (\ref{CaseII}) may be written as:
\[
D(g) = (24 K) h^2\,,\qquad D(h) = (24 K^{-1}) g^2 \,, \qquad D(\gamma) = -288\,g\,h\,.
\]
Here the weights of $g\, (=g_1)$ and $h\, (=g_4)$ are both 2.

\bigskip

\noindent{\underline{$N=4\,,\sigma=\frac{1}{3}$} The system (\ref{CaseIII}) may be written as:
\[
D(g_3) = \frac{3}{2} g_1^5\,,\qquad D(g_1) = g_3 \,, \qquad D(\gamma) = -g_1^4\,.
\]
Here the weights of $g_1$ and $g_3$ are 1 and 3 respectively.

\medskip

There is considerable scope for reducing such modular dynamical systems to canonical form by
redefining, if the weights allow, the functions $\gamma$ and $g_n\,.$

\begin{example}

Consider the system with $\cW=\{1\,,3\}\,.$ The most general system is
\begin{equation}
\begin{array}{rcl}
Dg_1 & = & a\, g_1^3 + b\, g_3 \,,\\
D\gamma & = & c \, g_1^4 + d \, g_1 g_3 \,,\\
D g_3 &=& e \, g_1^5 + f g_1^2 g_3\,
\end{array}
\label{13system}
\end{equation}
where $a\,,b\,,\ldots\,,f$ are arbitrary constants. Within this system one may make the following transformations
\begin{eqnarray*}
g_3 & \mapsto & g_3 + c_1 g_1^3\,,\\
\gamma & \mapsto & \gamma + c_2 g_1^2
\end{eqnarray*}
and with a careful analysis of the various (ignoring degenerate) cases one obtains the following canonical forms:

\[
\begin{array}{lclcl}
b =0\,: &\qquad& Dg_1 & = & 0 \,,\\
&& D \gamma &=& g_1 g_3 \,,\\
&& D g_3 &=& \alpha g_1^5 + \beta g_1^2 g_3\,,\\
&&&&\\
&&&&\\
b\neq 0\,: &\qquad& D g_1 & = & g_3\,,\\
&& D \gamma &=& -g_1^4 \,,\\
&& D g_3 &=& \alpha g_1^5 + \beta g_1^2 g_3\,,
\end{array}
\]
where $\alpha$ and $\beta$ are constants.

On eliminating $g_1$ and $g_3$ one obtains a third order modular ordinary differential equation for $\gamma$
from whose solution one can obtain $g_1$ and $g_3$
by quadrature. The third order system will not, in general, fall into Bureau's class, and the requirement that
it is in this class will place restrictions on the parameters.

For the $\cW=\{1\,,3\}$ system above:

\medskip

\begin{itemize}

\item[$\bullet$] If $b=0$ then it falls into Bureau's class if and only if $\alpha=0\,,\beta\neq 0$ or $\alpha\neq 0\,,\beta=0$
(if $\alpha=\beta=0$ the system is degenerate). The corresponding systems are, respectively
$(D\gamma)(D^3\gamma) = (D^2\gamma)^2$ and $D^3\gamma=0\,.$

\medskip

\item[$\bullet$] If $b\neq 0$ then it falls into Bureau's class if and only if $\beta=0$ and the corresponding system is
\[
(D\gamma)(D^3\gamma) = \frac{3}{4} (D^2\gamma)^2 - 4 \alpha (D\gamma)^3\,.
\]
The solution is given in terms of the Schwarzian triangle function $S[\frac{1}{2},\frac{1}{6},p,t]\,$ where
$p^2=\frac{1}{9}- \frac{1}{6\alpha }\,.$ Note that if one requires the global condition that the triangle functions
form a arithmetic triangle group then this forces $p=0$ and hence $\alpha=\frac{3}{2}\,$ which is precisely
the system \eqref{CaseIII}\,.

\end{itemize}

\end{example}

An alternative method to solve such modular dynamical system is to use a nonlinear change of dependent variable.
Again, for the $\{1,3\}$ system one may define:
\[
X = \gamma\,,\qquad Y = g_1^2\,,\qquad Z = \frac{g_3}{g_1}\,.
\]
With this the system (\ref{13system}) becomes
\begin{eqnarray*}
\dot{X} & = & X^2 + c Y^2 + d Y Z\,,\\
\dot{Y} & = & 2XY + 2 a Y^2 + 2 b YZ\,,\\
\dot{Z} & = & 2 XZ + e Y^2 + (f-a) YZ - b Z^2\,.
\end{eqnarray*}
This falls into the class of (three-component) quadratic dynamical systems considered by Ohyama \cite{Ohyama}:
\[
\dot{X}^i = \sum_{j,k} a^i_{jk} X^j X^k\,.
\]
In fact, all the above examples (except possibly the $\{1,1\}$-system) may be written, via appropriate nonlinear changes of
variable, in such a quadratic form. The method of solution depends on the commutative algebra defined by
\[
x_i \circ x_j = \sum_k a^k_{ij}x_k\,.
\]
In particular the modular invariance of the dynamical systems means that the corresponding algebra carries a unit. Once the algebra is identified (a complete
classification of three dimensional algebras with unity was carried out in \cite{Ohyama}) the solution is given in terms of solutions to second order Fuchsian
ordinary differential equations. In all examples appearing in this paper this turns out to be the hypergeometric equation.

We note also that in all existing examples the underlying modular dynamical system is of rank three. As dimension increases one might expect to find higher order
dynamical systems underlying the WDVV equations, but in all the examples so far found one just obtains third-order systems.

\section{Comments}

As mentioned in the introduction, some of the earliest examples of solutions of the WDVV equation \cite{VW,KTS} were actually examples of polynomial modular
Frobenius manifolds. Recently these have attracted renewed attention due to their appearance in orbifold quantum cohomology, i.e. the Gromov-Witten theory of
$\mathbb{P}^1_{a_1\,,\ldots\,,a_n}$ (the complex projective space with $n$-orbifold points). In particular one has an isomorphism (mirror symmetry)
between the orbifold picture and the singularity theory picture \cite{Rossi,ST}, e.g.
\begin{eqnarray*}
M_{\mathbb{P}^1_{2,2,2,2}} \simeq M_{D^{(1,1)}_4} \,,\\
M_{\mathbb{P}^1_{3,3,3}} \simeq M_{E^{(1,1)}_6} \,.\\
\end{eqnarray*}
The genus-1 contribution to the free energy may also be calculated for these examples \cite{iabs3}. In fact, the modularity property inherent in the genus-0
free energy also manifests itself at genus-1 \cite{iabs4}. Recently Dingdian and Zhang \cite{zhang} have conjectured how higher-genus correction should
transform under inversion symmetry. Ultimately one would wish to understand how the full dispersive hierarchy transforms under this modular symmetry: currently
only the genus-0 (or dispersionless) integrable systems transformation properties are known \cite{zhang, MS}.

\section*{Acknowledgements} IABS would like to thank the Carnegie Trust for the Universities of Scotland and SISSA,
and EM would like to thank the EPSRC, for financial support. Part of this paper was written during visits to SISSA (by IABS) and University of
Melbourne (by EM) and we would like to thank these institutions for their hospitality. We would also like to thank the referees for their
careful reading of the paper and their constructive comments.

\section*{Appendix}

In this appendix we summarize the results of Bureau \cite{Bureau1987} who constructed the general solution to the third order
differential equation
\begin{equation}
{\dddot{u}}= \alpha \frac{(\ddot{u}- 2 u \dot{u})^2}{\dot{u} - u^2} +\beta u \ddot{u} +\gamma \dot{u}^2 +\delta (\dot{u} - u^2)^2
\label{bureau}
\end{equation}
in terms of Schwarzian triangle functions. This will be referred to as Bureau's class.
Here the constants $\alpha\,,\beta\,,\gamma\,,\delta$ are constrained by the
relations
\[
16 \alpha - \gamma =18 \,,\quad 2 \beta + \gamma =6\,.
\]
The solution is given in terms of a function $y\,,$ namely
\[
u=\frac{1}{2}\left\{ \frac{\ddot{y}}{\dot{y}} - q(y) \dot{y}\right\}
\]
where
\[
q(y) = \left( \frac{(1-n)}{y} + \frac{(1-m)}{y-1}\right)
\]
and $y$ satisfies the Schwarzian differential equation
\begin{equation}
\frac{\dddot{y}}{\dot{y}} - \frac{3}{2} \left\{ \frac{\ddot{y}}{\dot{y}}\right\}^2 = -\frac{1}{2}
\left\{ \frac{1-n^2}{y^2} + \frac{1-m^2}{(1-y^2)} - \frac{(1+p^2 - m^2 - n^2)}{y(y-1)} \right\}\dot{y}^2\,.
\label{Schwarz}
\end{equation}
The relationship between $\alpha\,,\beta\,,\gamma\,,\delta$ and $m\,,n\,,p$ are given by the requirement that
$L_i=0\,,i=1\,,2\,,3$ where
\begin{eqnarray*}
L_1  & = & (1-2n) \left[ 1-3n - \alpha (1-2n) \right]\,,\\
L_2  & = & (1-2m) \left[ 1-3m - \alpha (1-2m) \right]\,,\\
L_3 & = & (1-2m)(2-3n) + (1-2n)(2-3m) - 2 \alpha (1-2m)(1-2n) \\
&&+ \frac{1}{4} (\gamma +\delta -6) \left[(1-m-n)^2 - p^2 \right]\,.
\end{eqnarray*}

As is well known, the solutions of the Schwarzian equation (\ref{Schwarz}) are given in terms of the Schwarzian triangle function
$y=S[m,n,p,t]$ which may be constructed as the inverse function of the ratio of two solutions to the hypergeometric
equation with parameters $a\,,b\,,c$ (i.e. one solution is ${}_2F_1[a,b,c,z]$) where
\[
n^2=(1-c)^2\,,\quad m^2=(a+b-c)^2\,,\quad p^2=(a-b)^2\,.
\]
Thus starting with a linear equation - the hypergeometric equation - one may, by taking ratios and inverse functions,
construct solutions to Bureau's equation (\ref{bureau}). In terms of the Rankin derivatives (\ref{bureau}) may be written
as
\[
(Du)(D^3u) = \alpha (D^2u)^2 + (\delta + 16 \alpha - 24) (Du)^3\,
\]
with $\beta = 12 - 8 \alpha\,,\gamma = 16 \alpha - 18\,.$ Note that
with these formulae:
\begin{eqnarray*}
\dot{u} - u^2  & = & h_1(y) \dot{y}^2\,,\\
& = &-\frac{1}{4} \left[ \frac{(1-m-n)^2 - p^2}{y(y-1)}\right] \dot{y}^2\,.
\end{eqnarray*}

In fact, Bureau's method is considerably simplified on the use of the Rankin derivative.
It is easy to show that
\[
D^nu = h_n(y) \dot{y}^{n+1}
\]
where $h_{n+1}(y)=\frac{d h_n(y)}{dy}  + (n+1) q(y) h_n(y)\,.$ Thus Bureau's equation (\ref{bureau}) becomes the algebraic equation
$h_1(y) h_3(y) = \alpha h_2(y)^2 + \beta h_1(y)^3\,$ which should hold for all values of $y\,.$
Equating coefficients gives the three equations $L_i=0\,.$ There
is considerable scope for generalizing this method by replacing the hypergeometric equation with more
general Fuchsian or even non-Fuchsian systems: again one just obtains algebraic conditions.

When $\alpha=0$ one obtains Chazy's equation class XII with solutions given in terms of the
function $S[\frac{1}{2},\frac{1}{3},p,t]$ where
$p^2 = \frac{1}{36} \left( \frac{\delta}{\delta - 24 }\right)$ and the Chazy equation itself occurs when, in addition,
$\delta=0$ with solution
given in terms of the function $S[\frac{1}{2},\frac{1}{3},0,t]\,.$

A systematic procedure to solve such modular differential equations has been derived in \cite{CO}\,. It is interesting to note, though,
that all the examples of modular Frobenius manifolds constructed here come from solutions of Bureau's equation and not more general
equations (or, more notably, higher-order modular equations).

\end{document}